%% file: main.tex
\documentclass[runningheads,envcountsame]{llncs}

\usepackage{etoolbox}
\newbool{lncs}
\newbool{final}
\booltrue{lncs}
\booltrue{final}

\input{preamble.tex}

\title{Approximate Multi-Matroid Intersection via Iterative Refinement%
}

\ifbool{lncs}{
\author{%
\addtocounter{footnote}{2}
Andr{\'e} Linhares%
\thanks{Research supported by NSERC grant 327620-09 and an NSERC DAS Award.}%
\fnmsep\inst{1} \and
Neil Olver%
\thanks{Supported by NWO VIDI grant 016.Vidi.189.087.}%
\fnmsep\inst{2} \and
\addtocounter{footnote}{-1}
Chaitanya Swamy$^{\text{\thefootnote}}$\fnmsep\inst{1} \and
\addtocounter{footnote}{1}
Rico Zenklusen%
\thanks{Supported by Swiss National Science Foundation grant 200021\_165866.}%
\fnmsep\inst{3}
}

\authorrunning{Linhares, Olver, Swamy, Zenklusen}

\institute{Dept. of Combinatorics and Optimization, Univ. Waterloo, Waterloo, Canada.
\email{\{alinhare,cswamy\}@uwaterloo.ca}
\and
Dept. of Econometrics \& Operations Research, Vrije Universiteit Amsterdam, Amsterdam, The Netherlands; CWI, Amsterdam, The Netherlands.
\email{n.olver@vu.nl}
\and
Department of Mathematics, ETH Zurich, Zurich, Switzerland.
\email{ricoz@math.ethz.ch}
}
}{
\author{%
Andr{\'e} Linhares\thanks{%
Department of Combinatorics and Optimization, University of Waterloo, Waterloo, Canada.
Email: \href{mailto:alinhare@uwaterloo.ca}%
{alinhare@uwaterloo.ca}.
}%
\and
Neil Olver\thanks{%
Department of Econometrics and Operations Research, Vrije Universiteit Amsterdam, Amsterdam, The Netherlands; and CWI, Amsterdam, The Netherlands.
Email: \href{mailto:n.olver@vu.nl}%
{n.olver@vu.nl}.
}%
\and
Chaitanya Swamy\thanks{%
Department of Combinatorics and Optimization, University of Waterloo, Waterloo, Canada.
Email: \href{mailto:cswamy@uwaterloo.ca}%
{cswamy@uwaterloo.ca}.
}%
\and 
Rico Zenklusen\thanks{%
Department of Mathematics, ETH Zurich, Zurich, Switzerland.
Email: \href{mailto:ricoz@math.ethz.ch}%
{ricoz@math.ethz.ch}.
Supported by Swiss National Science Foundation grant
200021\_165866.
}%
}
}

\begin{document}

\maketitle

\input{abstract.tex}

%

\input{introduction.tex}	
\input{mainRounding.tex}

\input{applications.tex}
\input{acknowledgments.tex}

\bibliographystyle{plain}
\bibliography{lit}

\newpage

\appendix

\input{largeAdditiveViolation.tex}

\section{Omitted proofs} \label{append-proofs}
\input{independenceVersion}

\input{appendix.tex}

\end{document}

%% file: preamble.tex
\ifbool{lncs}{

}{
\usepackage[letterpaper,margin=1in]{geometry}
\usepackage{times}
\usepackage{amsthm}
}

\usepackage{calc,comment}

\usepackage{bbold}
\usepackage{amsmath}
\usepackage{amssymb}
\usepackage{pgffor}
\usepackage[inline]{enumitem}

\setlist[enumerate]{nosep,topsep=0.3em}
\setlist[enumerate,1]{label=(\roman*)}
\setlist[enumerate,2]{label=(\alph*)}

\setlist[itemize]{nosep,topsep=0.1em}

\usepackage{mathtools}

\usepackage{bm}
\usepackage{microtype}
\usepackage{xspace}
\usepackage{xfrac}

\usepackage[margin=25pt,font=small,labelfont=bf]{caption}

\usepackage{array}
\usepackage[table]{xcolor}
\usepackage{multirow}
\usepackage{makecell}
\usepackage{booktabs}

\newcolumntype{M}[1]{>{\centering\arraybackslash}m{#1}}

\newcolumntype{L}[1]%
{>{\raggedright\let\newline\\\arraybackslash\hspace{0pt}}p{#1}}

\newcolumntype{C}[1]%
{>{\centering\let\newline\\\arraybackslash\hspace{0pt}}p{#1}}

\newcolumntype{R}[1]%
{>{\raggedleft\let\newline\\\arraybackslash\hspace{0pt}}p{#1}}

\newcolumntype{D}[1]%
{>{\raggedright\let\newline\\\arraybackslash\hspace{0pt}}m{#1}}

\newcolumntype{E}[1]%
{>{\centering\let\newline\\\arraybackslash\hspace{0pt}}m{#1}}

\newcolumntype{F}[1]%
{>{\raggedleft\let\newline\\\arraybackslash\hspace{0pt}}m{#1}}

\usepackage{url}
\usepackage{hyperref}

\hypersetup{colorlinks, linkcolor=darkblue, citecolor=darkgreen, urlcolor=darkblue}

\usepackage{float}
\usepackage{graphicx}

\graphicspath{{graphics/}}

\usepackage{tikz}

\usetikzlibrary{calc}
\usetikzlibrary{decorations.pathreplacing}
\usetikzlibrary{decorations.pathmorphing}
\usetikzlibrary{shapes.geometric}
\usetikzlibrary{fadings,decorations,fit}

\usepackage{tcolorbox}
\tcbuselibrary{theorems}
\tcbuselibrary{skins}

\newtcolorbox{probBox}[1][]{
enhanced,
colback=black!5,
attach boxed title to top left={xshift=2mm,yshift=-2mm},
#1
}

\usepackage[vlined,ruled,algo2e]{algorithm2e}

\usepackage[textsize=scriptsize]{todonotes}

\ifbool{final}
{
\newcommand\nnote[1]{}
\newcommand\rz[1]{}
}{
\newcommand\nnote[1]{{\color{blue}Neil: #1}}
\newcommand\rz[1]{{\color{blue}Rico: #1}}
}

\definecolor{darkblue}{rgb}{0,0,0.38}
\definecolor{darkred}{rgb}{0.6,0,0}
\definecolor{darkgreen}{rgb}{0.1,0.35,0}

\makeatletter
\newcommand{\labeltarget}[1]{\Hy@raisedlink{\hypertarget{#1}{}}}
\makeatother

\ifbool{lncs}{}{
\newtheorem{theorem}{Theorem}
\newtheorem{lemma}[theorem]{Lemma}

\newtheorem{corollary}[theorem]{Corollary}

}

\newlength\bxheight

\newcommand\OPT{\mathit{OPT}}
\newcommand\LPthreeMat{\ensuremath{\mathrm{LP}_{3\text{-mat}}}}
\newcommand\LPkMat{\ensuremath{\mathrm{LP}_{k\text{-mat}}}\xspace}
\newcommand\eqLPkMat{(\ensuremath{\mathrm{LP}_{k\text{-mat}}})\xspace}
\newcommand\LPMatchoid{\ensuremath{\mathrm{LP}_{\text{mat}}}}
\newcommand\LPKnapsack{\ensuremath{\mathrm{LP}_{\text{matkn}}}}

\DeclareMathOperator{\supp}{supp}

\DeclareMathOperator{\poly}{poly}

\newcommand{\APX}{\textsl{APX}\xspace}
\newcommand{\Pcomp}{\textsl{P}\xspace}
\newcommand{\NPcomp}{\textsl{NP}\xspace}
\newcommand{\p}{\Pcomp}
\newcommand{\np}{\NPcomp} 
\newcommand{\nphard}{\np-hard\xspace}

\newcommand{\espace}{\,}

\newcommand{\R}{\ensuremath{\mathbb R}} 
\newcommand{\Z}{\ensuremath{\mathbb Z}}
 
\newcommand{\A}{\ensuremath{\mathcal{A}}}
\newcommand{\C}{\ensuremath{\mathcal{C}}} 

\newcommand{\I}{\ensuremath{\mathcal I}}

\newcommand{\M}{\ensuremath{\mathcal M}}

\newcommand{\sm}{\ensuremath{\setminus}} 
\newcommand{\es}{\ensuremath{\emptyset}}

\newcommand{\sse}{\subseteq} 

\newcommand{\assign}{\ensuremath{\leftarrow}}

\newcommand{\e}{\ensuremath{\epsilon}}

\newcommand{\al}{\ensuremath{\alpha}} 
\newcommand{\dt}{\ensuremath{\delta}} 
\newcommand{\Dt}{\ensuremath{\Delta}} 

\newcommand{\qedhere}{\tag*{\qed}}
\newcommand{\bg}[1]{\medskip\noindent{\it {#1}.\ }}
\newcommand{\ed}{{\hfill\qed}\medskip} 

\newenvironment{proofof}[1]{\bg{Proof of {#1}}}{\ed}

\newenvironment{proofsketchofnobox}[1]{\bg{Proof sketch of {#1}}}{\hfill\medskip}

\spnewtheorem*{3matprob}{Weighted 3-matroid intersection}{\bfseries}{\rmfamily}

\newcommand{\GMDBMST}{\textsc{gmdst}\xspace}
\newcommand{\gmdst}{\GMDBMST}

\newcommand{\base}{\mathcal{B}}

%% file: abstract.tex
\begin{abstract}
We introduce a new iterative rounding technique to round a point in a matroid polytope
subject to further matroid constraints. This technique returns an independent set in one
matroid with limited violations of the other ones. On top of the classical steps of
iterative relaxation approaches, we iteratively refine/split involved matroid constraints
to obtain a more restrictive constraint system, that is amenable to iterative relaxation
techniques. Hence, throughout the iterations, we both tighten constraints and later relax
them by dropping constrains under certain conditions. 
Due to the refinement step, we can deal with considerably more general constraint classes
than existing iterative relaxation/rounding methods, which typically round on one matroid
polytope with additional simple cardinality constraints that do not overlap too much. 

We show how our rounding method, combined with an application of a matroid intersection
algorithm, yields the first $2$-approximation for finding a maximum-weight
common independent set in $3$ matroids. Moreover, our $2$-approximation is LP-based, and
settles the integrality gap for the natural relaxation of the problem. Prior to our work,
no better upper bound than $3$ was known for the integrality gap, which followed from the
greedy algorithm. 
\nnote{Instead of the last sentence, what about:}
We also discuss various other applications of our techniques, including an extension that
allows us to handle a mixture of matroid and knapsack constraints.
\end{abstract}

%% file: introduction.tex
\section{Introduction} \label{intro}

Matroids are among the most fundamental and well-studied structures in combinatorial optimization.
Recall that a {\em matroid} $M$ is a pair $M=(N,\mathcal{I})$, where $N$ is a finite ground set and
$\mathcal{I}\subseteq 2^N$ is a family of sets, called \emph{independent sets}, such that 
\begin{enumerate*}
\item $\es\in\I$,
\item if $A\in \mathcal{I}$ and $B\subseteq A$, then $B\in \mathcal{I}$, and
\item if $A,B\in \mathcal{I}$ with $|A| > |B|$, then there is an element 
$e\in A\setminus B$ such that $B\cup \{e\}\in \mathcal{I}$. 
\end{enumerate*}
As is common when dealing with matroids, 
we assume that a matroid is specified via an {\em independence oracle}, that, given $S\sse N$ as input, returns if $S\in\I$.
Matroids capture a variety of interesting problems, and matroid-optimization algorithms provide a
powerful tool in the design and analysis of efficient algorithms. A key
matroid-optimization problem 
is the {\em matroid intersection} problem, wherein  
we seek a maximum-weight set that is independent in {\em two} matroids, for which various
efficient algorithms are known, and we also have a celebrated min-max theorem and a polyhedral 
understanding of the problem.
The versatility of matroid intersection comes from the fact that the intersection of
matroids allows for describing a very broad family of constraints. 

Unfortunately, as soon as the intersection of $3$ or more matroids is considered, already the unweighted version of determining a maximum cardinality common independent set becomes \APX-hard. Due to its fundamental nature, and many natural special cases, the problem of optimizing over $3$ or more matroids has received considerable attention. 
In particular, there is extensive prior work ranging from the study of maximum cardinality
problems~\cite{lau_2011_iterative}, the maximization of submodular functions over the
intersection of multiple matroids~(see~\cite{fisher_1978_analysis,lee_2010_maximizing,%
lee_2010_submodular,gupta_2010_constrained,chekuri_2014_submodular} and the references therein), to various interesting special cases like $k$-dimensional matching~(see~\cite{hurkens_1989_size,halldorsson_1995_approximating,%
chan_2012_linear,cygan_2013_how,cygan_2013_improved} and the references therein; many of
these results apply also to the $k$-set packing problem which generalizes $k$-dimensional matching).

Nevertheless, there are still basic open questions regarding the approximability of the
optimization over $3$ or more matroids. 
Perhaps the most basic problem of this type is the 
\emph{weighted $3$-matroid intersection} problem, 
defined as follows. 

\begin{3matprob} 
Given matroids $M_i=(N,\I_i)$, for $i=1,2,3$, on a common ground set $N$, and a weight
vector $w\in\R^N$, solve
\begin{equation*}
\max\ \left\{ w(I):\ I\in \mathcal{I}_1 \cap \mathcal{I}_2 \cap \mathcal{I}_3 \right\},
\end{equation*}
where we use the shorthand $w(S)\coloneqq \sum_{e\in S} w(e)$ for any set $S\subseteq N$. 
\end{3matprob}

The \emph{unweighted $3$-matroid intersection} problem, which is also sometimes called
the {\em cardinality} version of 3-matroid intersection, 
is the special case where $w(e)=1$ for all $e\in N$, so $w(S)=|S|$ for $S\subseteq N$.   

The $3$-matroid intersection problem 
has a natural 
and canonical LP-relaxation: 
\begin{equation}\label{eq:3MatLP}
\max\ \left\{ w^T x:\ x \in P_{\mathcal{I}_1} \cap P_{\mathcal{I}_2} \cap P_{\mathcal{I}_3}\right\},\tag{\LPthreeMat}
\end{equation}
where, for a matroid $M=(N,\mathcal{I})$, we denote by $P_{\mathcal{I}}\subseteq [0,1]^N$ the matroid polytope of $M$, which is the convex hull of all characteristic vectors of sets in $\mathcal{I}$. It has a well known inequality description given by 
\begin{equation*}
P_{\mathcal{I}}  = \bigl\{ x\in \mathbb{R}_{\geq 0}^N:\ x(S)\leq r(S) \;\;\forall S\subseteq N \bigr\}\espace,
\end{equation*}
where $r:2^N \longrightarrow \mathbb{Z}_{\geq 0}$ is the \emph{rank function} of $M$,
which, for $S\subseteq N$, is defined by
$r(S) \coloneqq \max\bigl\{|I|:\ I\in \mathcal{I}, I\subseteq S\bigr\}$.
The rank function is submodular, and $r(S)$ can be computed for any $S\sse N$ using an
independence oracle. It will therefore often be convenient to assume that a matroid $M$
is specified via its {\em rank oracle} that given $S\sse N$ as input, returns $r(S)$. In
particular, 
one can efficiently optimize any linear function over $P_{\I}$ given a
rank oracle (or equivalently an independence oracle).
The above LP-relaxation extends naturally to the {\em $k$-matroid intersection} problem,
which is the extension of $3$-matroid intersection to $k$ matroids. 

Whereas \eqref{eq:3MatLP}, and its extension \eqLPkMat to $k$-matroid intersection, are
well-known LP-relaxations, there remain various gaps in our understanding of these
relaxations. 
It is widely known that the greedy algorithm is a $k$-approximation for $k$-matroid
intersection. 
Moreover, this approximation is relative to the optimal value of \eqLPkMat, 
which leads to the current-best upper bound of $k$ on the integrality gap of \eqLPkMat,
for all $k\geq 3$. 
However, the best lower bound on the integrality gap of~\eqLPkMat is $k-1$ (also for all
$k\geq 3$);  
this is known to be achievable in instances where the involved matroids are partition
matroids, and for unweighted instances~\cite{furedi_1981_maximum,lau_2011_iterative}. 

Significant progress on approximating $k$-matroid intersection was achieved by Lee,
Sviridenko, and Vondr{\'a}k~\cite{lee_2010_submodular}, who presented, for any fixed
$\epsilon >0$, a local search procedure with running time exponential in $\epsilon$ that
leads to a ${k-1+\epsilon}$-approximation (i.e., the weight of the set returned is at
least (optimum)/$(k-1+\e)$). Unfortunately, apart from its high
running time dependence on $\epsilon$, this approach does not shed any insights on
\eqLPkMat, as the above guarantee is not relative to $\OPT_{\LPkMat}$.
%
Further progress on understanding the quality of the LP-relaxations has only been achieved in
special cases. In particular, for {\em unweighted $k$-matroid intersection}, Lau, Ravi and
Singh~\cite{lau_2011_iterative} give an LP-based $(k-1)$-approximation through iterative
rounding. 
Their proof is based on identifying an element with ``large''
fractional value, 
picking it, and altering the fractional solution
so that it remains feasible; the last step crucially uses the fact that the instance
is unweighted to control the loss in the LP objective value. 
For the intersection of $k$ partition matroids, a problem also known as
\emph{$k$-dimensional matching}, Chan and Lau~\cite{chan_2012_linear} were able to obtain
a $(k-1)$-approximation based on \eqLPkMat. 

Although it is generally believed that a $(k-1)$-approximation for $k$-matroid
intersection should exist, and that the integrality gap \eqLPkMat is equal to the known
lower bound of $k-1$, this has remained open even for $3$-matroid intersection 
(prior to our work). 
Recall that in this case, the best-known upper and lower bounds on the integrality gap
of \eqref{eq:3MatLP} are $3$ (via the classical greedy algorithm) and $2$ respectively.
Moreover, the only method to beat the trivial $3$-approximation of the greedy
algorithm is the non-LP based and computationally quite expensive
$(2+\e)$-approximation in~\cite{lee_2010_submodular}. 
One main reason for the limited progress is the lack of techniques for rounding points in the
intersection of multiple matroid polytopes with sufficiently strong properties. In particular, one
technical difficulty that is encountered is that the tight constraints (even at an
extreme point) may have large overlap, and we do not know of ways for dealing with this.

\vspace{-0.5em}
\subsubsection*{Our results.}
We introduce a new iterative-rounding approach to handle the above difficulties, that
allows for dealing with a very general class of optimization problems involving
matroids. Before delving into the details of this 
technique, we highlight its main implication in the context of $3$-matroid intersection. 
\begin{theorem}\label{thm:3MatInt} \label{3mat}
There is an LP-relative $2$-approximation for weighted $3$-matroid intersection.
That is, for any instance, we can efficiently find a common independent set $R$ with
$w(R)\geq\sfrac{\OPT_{\text{\ref{eq:3MatLP}}}}{2}$; 
thus, the integrality gap of \eqref{eq:3MatLP} is at most $2$.
\end{theorem}
This is the {\em first} $2$-approximation for $3$-matroid intersection (with general weights).  
Moreover, our result {\em settles} the integrality gap of~\eqref{eq:3MatLP} due to the known
matching integrality-gap lower bound of $2$. 

The chief new technical ingredient that leads to Theorem~\ref{3mat}, and results for
other applications discussed in Section~\ref{apps}, is an approximation result
based on a novel \emph{iterative refinement} technique (see Section~\ref{refine})
for problems of the following type. 
%
Let $N=N_0$ be a finite ground set, and $M_i=(N_i,\mathcal{I}_i)$ for $i=0,\ldots,k$ be $k+1$ matroids with rank functions $\{r_i\}$, where $N_i\subseteq N$ and $w\in\R^N$ be a weight
vector (note that {\em negative} weights are allowed). 
We consider the problem  
\begin{equation}\label{eq:baseQMatchoid}
    \max\ \bigl\{w(I): \ I \in \base_0, \quad I \cap N_i \in \mathcal{I}_i \ \ \forall i \in [k] \bigr\},
\end{equation}
where $\base_0$ is the set of all bases of $M_0$ and $[k] := \{1, \dots, k\}$.
%
The reason we consider matroids $M_i$ for $i\in [k]$ defined on ground sets $N_i$ that are
subsets of $N$, is because, as we show below, we obtain guarantees depending on how
strongly the sets $N_i$ overlap; intuitively, problem \eqref{eq:baseQMatchoid} becomes
easier as the overlap between $N_1,\ldots,N_k$ decreases, and our guarantee improves
correspondingly. 

We cannot hope to solve \eqref{eq:baseQMatchoid} optimally,
as this would enable one to solve the \nphard $k$-matroid intersection problem. 
Our goal will be to  find a basis of $M_0$ of large weight 
that is ``approximately independent'' in the matroids
$M_1,\ldots,M_k$. 

How to quantify ``approximate independence''?
Perhaps the two notions that first come to mind are additive 
and multiplicative violation of the rank constraints. 
Whereas additive violations are common in the study of degree-bounded MST problems, which
can be cast as special cases of \eqref{eq:baseQMatchoid}, it turns out that such a guarantee is
impossible to obtain (in polytime) for \eqref{eq:baseQMatchoid}. 
More precisely, we show in Appendix~\ref{addviol} (via a replication idea) that, even for
$k=2$, if we could find in polytime a basis $B$ of $M_0$ 
satisfying $|B|\leq r_i(B)+\al$ for $i=1,2$ for 
$\al = O(|N|^{1-\epsilon})$ for any $\epsilon > 0$,
then we could efficiently find a basis of $M_0$ that is independent in $M_1$, $M_2$; the latter problem is easily seen to be \nphard via a reduction
from Hamiltonian path. 
\nnote{I don't think sublinear was precisely correct here.}
We therefore consider multiplicative violation of the rank constraints. We say that
$S\sse N$ is {\em $\al$-approximately independent}, or simply {\em $\al$-independent}, for 
a matroid $M=(N,\I)$, if  
$|T|\leq \al\cdot r(T)\ \forall T\sse S$ (equivalently, $\chi^S\in\al P_{\I}$, where
$\chi^S$ is the characteristic vector of $S$). 
This is much stronger than simply requiring that $|S|\leq\al\cdot r(S)$, and
it is easy to give examples where this weaker notion admits sets that one would consider
to be quite far from being independent. 
An appealing feature of the stronger definition is that, using the
min-max result for matroid-intersection (or via matroid partition; see, e.g.,
\cite{CookCPS}), it follows easily that if $\al\in\Z_{\geq 0}$, then $S$ is
$\al$-independent iff $S$ can be partitioned into at most $\al$ independent sets of $M$. 
We now state the guarantee we obtain for \eqref{eq:baseQMatchoid} precisely.
\nnote{I wonder if we should change the label. We don't use matchoid anywhere in the text at all, so it's a bit weird for people who aren't familiar with the term. Plus it's long.}
We consider the following canonical LP-relaxation of~\eqref{eq:baseQMatchoid}:
%
\begin{equation}\label{eq:mainLP}
\max\ \left\{w^Tx: \ \  
x \in \R_{\geq 0}^N, \quad x \in P_{\base_0}, \quad x\vert_{N_i}\in P_{\I_i} \ \ \forall i \in [k] \right\},
\tag{\LPMatchoid}
\end{equation}
%
where for a set
$S\subseteq N$, we use $x\vert_{S}\in \mathbb{R}^S$ to denote the restriction of $x$ to
$S$.  
For ease of notation, we will sometimes write $x\in P_{\mathcal{I}_i}$ and $R\in\I_i$ instead of $x\vert_{N_i} \in P_{\mathcal{I}_i}$ and $R\cap N_i\in\I_i$, respectively.
Our main result 
for~\eqref{eq:baseQMatchoid}, based on a new iterative rounding algorithm
for~\eqref{eq:mainLP} described in Section~\ref{sec:mainRounding}, 
is the following.  
\begin{theorem}\label{thm:mainThm} \label{thm:mainthm}
Let $q_1,\ldots, q_k \in \mathbb{Z}_{\geq 1}$ such that 
\begin{equation}\label{eq:loadProp}
\sum_{i\in [k]: e\in N_i} q_i^{-1} \leq 1 \qquad \forall e\in N\espace.
\end{equation}
If \eqref{eq:mainLP} is feasible, then one can efficiently compute $R \subseteq N$ such that
\begin{enumerate}
\item $R\in\base_0$;
\item $w(R)\geq\OPT_{\text{\ref{eq:mainLP}}}$; and 
\item $R$ is $q_i$-independent in $M_i$\ $\forall i \in [k]$.
\end{enumerate}
%
\end{theorem}
Note that, in particular, taking $q_i=\max_{e\in N}\bigl|\{j\in[k]: e\in N_j\}\bigr|$ for
all $i \in [k]$ satisfies \eqref{eq:loadProp}. 
%
%
Thus, 
we violate the constraints imposed by the other matroids
$M_1,\ldots, M_k$ by a multiplicative factor depending on how strongly the $N_i$s
overlap. 

While we have stated Theorem~\ref{thm:mainThm} in terms of \emph{bases} of $M_0$, the following natural variant is easily deduced from it (we defer the proof to the Appendix~\ref{append-proofs}).

\begin{corollary}\label{cor:mainInd}
Theorem~\ref{thm:mainThm} also holds when $R$ is required only to be an independent set in $M_0$ 
(as opposed to a basis), 
and \eqref{eq:mainLP} is replaced by 
\begin{equation}\label{eq:indLP}
\max\ \bigl\{w^Tx: \ \  
x \in \R_{\ge 0}^N, \quad x\vert_{N_i}\in P_{\I_i} \ \ \forall i= 0,1,\ldots,k \bigr\}\espace.
\end{equation}
\end{corollary}
A variety of problem settings can be handled via Theorem~\ref{thm:mainThm} and Corollary~\ref{cor:mainInd}
in a unified way. We first show how to obtain a crisp, simple proof of
Theorem~\ref{thm:3MatInt}.

\begin{proofof}{Theorem~\ref{thm:3MatInt}}
Given matroids $M_i = (N,\mathcal{I}_i)$ for $i=0,1,2$, and a weight vector $w \in
\mathbb{R}^N$, we first solve~\eqref{eq:3MatLP} to obtain an optimal
solution $x^*$. Now we utilize Corollary~\ref{cor:mainInd} with the same three matroids, and $q_1=q_2=2$. 
Clearly, these $q$-values satisfy
\eqref{eq:loadProp}, and $x^*$ is a feasible solution to \eqref{eq:indLP}.
Thus we obtain a set $A\in\I_0$ with 
$w(A)\geq w^Tx^*$ and 
$\chi^A\in 2P_{\I_1}\cap 2P_{\I_2}$.

It is well known that $P_{\I_1}\cap P_{\I_2}$ is a polytope with integral extreme
points (see, e.g.,~\cite{CookCPS}). So since $\chi^A/2\in P_{\I_1}\cap P_{\I_2}$, by using
an algorithm for (weighted) matroid intersection applied to matroids $M_1$ and $M_2$
restricted to $A$, we can find a set $R\sse A$ such that $R\in\I_1\cap\I_2$ and 
$w(R)\geq w^T\chi^A/2\geq w^T x^*/2$. Finally, since $R\sse A$ and
$A\in\I_0$, we also have that $R\in\I_0$.
\end{proofof}

Beyond 3-matroid intersection, Theorem~\ref{thm:mainthm} is applicable to various 
constrained (e.g., degree-bounded) spanning tree problems; we expand on this below. 
In Section~\ref{apps}, we discuss an application in this
direction, wherein we seek a min-cost spanning tree 
satisfying matroid-independence constraints on the edge-sets of a given disjoint
collection of node sets. Using Theorem~\ref{thm:mainthm}, we obtain a spanning tree with a
multiplicative factor-2 violation of the matroid constraints.

In Section~\ref{apps}, we also present a noteworthy extension of
Theorem~\ref{thm:mainthm} with $t$ \emph{knapsack constraints} in addition to $k$ matroid
constraints, and show that we can 
obtain multiplicative violations of both the matroid and knapsack constraints. 
The only other such result we are aware of that applies to a \emph{mixture} of matroid
and knapsack constraints is by Gupta et al.~\cite{GuptaNR}; their result 
in our setting yields an $O(kt)$-approximation with no constraint violation, which
is incomparable to our result.

\vspace{-1.0em}

\subsubsection*{Related work and connections.}
We note that by choosing $M_0$ to be the graphic matroid,
problem \eqref{eq:baseQMatchoid} generalizes a variety of known
constrained spanning tree problems. This includes 
degree-bounded spanning trees, and generalizations thereof considered by Bansal et al.~\cite{bansal_2009_additive}, Kir\'aly et al.~\cite{kiraly_2008_degree}, and Zenklusen~\cite{zenklusen_2012_matroidal}. 
Theorem~\ref{thm:mainThm} thus yields a unified way to deal with various spanning tree problems considered in the literature, where
the soft/degree constraints are violated by at most a constant factor. 
However, as noted earlier, whereas the above works obtain stronger, additive-violation
results, for the various constrained spanning tree problems they consider, such guarantees
are not possible for our general problem \eqref{eq:baseQMatchoid} 
(see Appendix~\ref{addviol}). This hardness (of obtaining small additive violations) carries
over to the spanning tree application that we consider in Section~\ref{apps} (which
generalizes the matroidal degree-bounded spanning tree problem considered
by~\cite{zenklusen_2012_matroidal}).   

To showcase how Theorem~\ref{thm:mainThm} can be used for such problems, consider the
minimum degree-bounded spanning tree problem, where given is a graph $G=(V,E)$ with edge
weights $w:E\rightarrow \mathbb{R}$ and degree bounds $B_v\in \mathbb{Z}_{\geq 1}$ for
$v\in V$. The nominal problem asks to find a spanning tree $T\subseteq E$ with $|T\cap
\delta(v)|\leq B_v$ for $v\in V$ minimizing $w(T)$, where $\delta(v)$ denotes the set of edges incident with $v$. Here one can apply
Theorem~\ref{thm:mainThm} with $M_0$ being the graphic matroid of $G$, and for each
$v\in V$ we define a uniform matroid $M_v$ with ground set $\delta(v)$ and rank
$B_v$. Theorem~\ref{thm:mainThm} with $q_v=2 \;\forall v\in V$ and negated edge weights leads to a spanning tree $T$ with $|T\cap \delta(v)|\leq 2B_v \; \forall v\in V$ and weight no more than the optimal LP-weight. Whereas this is a simple showcase example, Theorem~\ref{thm:mainThm} can be used in a
similar way for considerably more general constraints than just degree constraints. 

Finally, we highlight a main difference of our approach compared to prior
techniques. Prior techniques for related problems, as used for example by Singh and
Lau~\cite{singh_2007_approximating}, Kir{\'a}ly et al.~\cite{kiraly_2008_degree}, 
and Bansal et al.~\cite{bansal_2009_additive}, 
successively drop constraints of a relaxation.  
Also, interesting variations have been suggested that do not just drop constraints but may
relax constraints by replacing a constraint by a weaker family (see work by Bansal et
al.~\cite{bansal_2010_generalizations}). 
%
In contrast, 
our method does not just relax constraints, but also
strengthens the constraint family in some iterations, so as to simplify it and enable one
to drop constraints later on. 

%% file: mainRounding.tex
\section{Our rounding technique }\label{sec:mainRounding} \label{round} 
\label{refine}

Our rounding technique heavily relies on a simple yet very useful ``splitting'' procedure
for matroids, which we call \emph{matroid refinement}.

\subsubsection*{Matroid refinement.}

Let $M=(N,\mathcal{I})$ be a matroid with rank function $r:2^N\rightarrow \mathbb{Z}_{\geq 0}$, and let $S\subsetneq N$, $S\neq \emptyset$. The
\emph{refinement} of $M$ with respect to $S$ are the two matroids $M_1=M\vert_{S}$
obtained by restricting $M$ to $S$, and $M_2=M/S$ obtained by contracting $S$ in $M$. 
Formally, the independent sets of the two matroids
$M_1=(S,\mathcal{I}_1), M_2=(N\setminus S, \mathcal{I}_2)$ are given by 
$\mathcal{I}_1=\left\{I\subseteq S: I\in \mathcal{I} \right\}$, and
$\mathcal{I}_2= \left\{I\subseteq N\setminus S: I\cup I_S \in \mathcal{I} \right\}$,
where $I_S\in \mathcal{I}$ is a maximum cardinality independent subset of $S$. It is
well-known that the definition of $\mathcal{I}_2$ does not depend on which set $I_S$ is
chosen. The rank functions $r_1:2^S\rightarrow \mathbb{Z}_{\geq 0}$ and
$r_2:2^{N\setminus S} \rightarrow \mathbb{Z}_{\geq 0}$ of $M_1$ and $M_2$, respectively,
are given by 
%
\begin{equation}\label{eq:rankRefinements}
r_1(A) = r(A) \;\;\forall A\subseteq S \text{ , and } \qquad r_2(B) = r(B\cup S) -r(S) \;\;\forall B\subseteq N\setminus S\enspace.
\end{equation}

We refer the reader to~\cite[Volume B]{schrijver_2003_combinatorial} for more information
on 
matroid restrictions and contractions. 
The following lemma describes two basic yet important relations between a matroid
$M=(N,\mathcal{I})$ and its refinements $M_1=M\vert_S$ and $M_2=M/S$. 
These relations easily follow from well-known properties of
matroids; 
we include the proofs in Appendix~\ref{append-proofs} for completeness.
\nnote{I think we can safely omit the proofs from the submission.}

\begin{lemma}\label{lem:refinement}
Let $x\in \mathbb{R}^N$ such that $x\vert_S\in P_{\mathcal{I}_1}$ and 
$x\vert_{N\sm S}\in P_{\I_2}$. Then  $x\in P_{\mathcal{I}}$. 
\end{lemma}

\begin{lemma}\label{lem:remFeasible}
Let $x\in P_{\mathcal{I}}$ such that $x(S)=r(S)$. Then $x\vert_{S}\in P_{\mathcal{I}_1}$
and $x\vert_{N\sm S}\in P_{\I_2}$.
\end{lemma}

%
%

Intuitively, the benefit of matroid refinement is that it serves to partly decouple the
matroid independence constraints for $M$, thereby allowing one to work with somewhat
``simpler'' matroids subsequently, and we leverage this carefully in our algorithm.

\subsubsection*{An algorithm based on iterative refinement and relaxation.}

Algorithm~\ref{alg:mainAlg} describes our method to prove
Theorem~\ref{thm:mainThm}. Recall that the input is an instance of problem
\eqref{eq:baseQMatchoid}, which consists of $k+1$ matroids  
$M_i=(N_i,\mathcal{I}_i)$ for $i=0,\ldots,k$, where each $N_i$ is a subset of a finite
ground set $N=N_0$, and a weight vector $w\in \mathbb{R}^N$. 
We are also given integers $q_i\geq 1$ for $i\in [k]$ satisfying~\eqref{eq:loadProp}. 

\begin{algorithm2e}
\begin{enumerate}[nosep,leftmargin=0.5em,label=\arabic*.,ref=\arabic*]
    \item\label{item:init} Initialize $\M\assign\{M_1,\ldots,M_k\}$, $q_{M_i} \gets q_i$ for all $i \in [k]$.
\item\label{item:solveLP} Compute an optimal basic
solution $x^*$ to~\eqref{eq:mainLP} for the matroids $\{M_0\}\cup\M$.

\item\label{item:delcontr} 
    Delete all $e\in N$ with $x^*(e)=0$ and contract all $e \in N$ with $x^*(e)=1$ from all relevant matroids, updating also the ground set $N$. 

\item\label{item:exit1}
If $N=\emptyset$: \textbf{return} the set of all elements contracted so far.

\item\label{item:refine} \textbf{While} there is a matroid $M'=(N',\I')\in\M$ 
with associated rank function $r'$, s.t. $\exists\,\es\neq S\subsetneq N'$ with
$x^*(S)=r'(S)$: 
\hspace*{0.5cm}\parbox[t]{0.9\linewidth}{
    \emph{(Refinement.)} Set $M_1' = M'|_S$, $M_2' = M' / S$, and $q_{M_1'} = q_{M_2'} = q_{M'}$.\\
    Update $\M\assign(\M\sm\{M'\})\cup\{M'_1, M'_2 \}$.
}

\item\label{item:dropMat}
Find a matroid $M'=(N',\I')\in\M$ with associated rank function $r'$, 
such that $x^*(N')=r'(N')$ and $|N'| - x^*(N')<q_{M'}$; remove $M'$ from $\M$.
Go to step~\ref{item:solveLP}.
\end{enumerate}

\caption{Iterative refinement/relaxation algorithm
for Theorem~\ref{thm:mainThm}}
\label{alg:mainAlg}
\end{algorithm2e}

Algorithm~\ref{alg:mainAlg} starts by solving the natural LP-relaxation in
step~\ref{item:solveLP} to obtain a point $x^*$. As is common in iterative rounding
algorithms, we delete all elements of value $0$ and fix all elements of value $1$ through
contractions in step~\ref{item:delcontr}. 
Apart from these standard operations, we {\em refine} the matroids
in step~\ref{item:refine}, as long as there is a nontrivial $x^*$-tight set in some
matroid in our collection. Notice that after step~\ref{item:refine}, the $q$-values for
the matroids in the new collection $\M$ continue to satisfy \eqref{eq:loadProp}.
Step~\ref{item:dropMat} is our relaxation step, where we
drop a matroid $M'=(N',\mathcal{I}')$ if $|N'|-x^*(N') < q_{M'}$. This is the step that
results in a violation of the matroid constraints, but, as we show, the above condition
ensures that even if we select all elements of $N'$ in the solution, the violation is
still within the prescribed bounds.  

\nnote{Modified slightly.}
In order to find an $x^*$-tight set $\es\neq S\subsetneq N'$ (if one
exists) in step~\ref{item:refine}, 
one can, for example, minimize the submodular function $r'(A)-x^*(A)$ over the sets
$\emptyset\neq A\subsetneq N'$. Depending on the matroids involved, faster specialized
approaches can be employed. 

It is perhaps illuminating to consider the combined effect of all the refinement steps and  
step~\ref{item:dropMat} corresponding to a given basic optimal solution $x^*$.
Using standard uncrossing techniques, one can show that for each matroid
$M'=(N',\I')\in\M$, there is a nested family of sets 
$\es\subsetneq S_1\subsetneq\ldots\subsetneq S_p\sse N'$ whose rank constraints span the
$x^*$-tight constraints of $M'$, and so any $S_i$ can be used to refine $M'$.
The combined effect of steps~\ref{item:refine} for $M'$ can be seen as replacing $M'$ by
the matroids $\bigl(M'\vert_{S_\ell}\bigr)/S_{\ell-1}$ 
for $\ell=1,\ldots,p+1$, where $S_0\coloneqq\es$, $S_{p+1}\coloneqq N'$. 
Step~\ref{item:dropMat} chooses some $M'\in\M$ and a ``ring'' $S_\ell\sm S_{\ell-1}$
of its nested family satisfying $|S_\ell\sm S_{\ell-1}|-x^*(S_\ell\sm S_{\ell-1})<q_{M'}$, 
and drops the matroid created for this ring.

\paragraph{Analysis.}
Lemma~\ref{lem:approx} shows that if Algorithm~\ref{alg:mainAlg} terminates, then it
returns a set with the desired properties. In Lemma~\ref{lem:terminate}, we show that the
algorithm terminates in a polynomial number of iterations. In particular, we show that in
step~\ref{item:dropMat}, there will always be a matroid in our collection that we can
drop. 

\begin{lemma} \label{lem:approx}
Suppose that Algorithm~\ref{alg:mainAlg} returns a set $R\subseteq N$. Then, $R$ satisfies
the properties stated in Theorem~\ref{thm:mainThm}. 
\end{lemma}

\begin{proof}
Note that $R\in\base_0$, as $M_0$ is only modified via
deletions or contractions. 
Moreover, $w(R)\geq\OPT$, where $\OPT$ is the
optimal value of \eqref{eq:mainLP} for the input instance. Indeed, if $x^*$ is the
current optimal solution, and we update our instance (via deletions, contractions,
refinements, or dropping matroids), then $x^*$ restricted to the new ground set remains feasible for~\eqref{eq:mainLP} for the new instance. This is immediate for
deletions and contractions, and if we drop a matroid; it holds for refinements due to 
Lemma~\ref{lem:remFeasible}. So if the optimal value of \eqref{eq:mainLP}
decreases, this is only because we contract elements with $x^*(e)=1$, which we include in
$R$. It follows that $w(R)\geq\OPT$. 

We now show that $R$ is $q_i$-independent in $M_i$ for all $i\in[k]$.
Consider the state of the algorithm at a point during its execution right before
performing step~\ref{item:solveLP}. Hence, the instance may already have been modified
through prior refinements, contractions, deletions, and relaxations. 
We claim that the following invariant holds throughout the algorithm: 

\begin{center}
\parbox[t]{0.95\linewidth}{
If $R'$ satisfies the properties of Theorem~\ref{thm:mainThm} with respect to
the current instance, then the set $R$ consisting of $R'$ and all elements contracted so
far fulfills the properties of Theorem~\ref{thm:mainThm} with respect to the original
instance. 
}
\end{center}

\nnote{One thing I found confusing here is that in the claim, $R$ refers to $R'$ with all contracted elements added; 
but in what follows, $R$ is actually just $R'$ in the previous iteration. Maybe there's a slightly clearer way to do this?}
To show the claim, it suffices to show that the invariant is preserved whenever we change
the instance in the algorithm.
First, one can observe that if the instance changes by deleting an element of value $0$
or contracting an element of value $1$, then the invariant is preserved. 
Next, consider step~\ref{item:refine}, where we refine $M'=(N',\mathcal{I}')\in\M$ to
obtain $M'\vert_S=(S,\I'_1)$ and $M'/S=(N'\sm S,\I'_2)$ whose $q$-values are set to
$q_{M'}$. We are given that $\chi^{R'}\vert_{S}\in q_{M'}P_{\I'_1}$ and 
$\chi^{R'}\vert_{N'\sm S}\in q_{M'}P_{\I'_2}$. 
So by Lemma~\ref{lem:refinement}, we have $\chi^{R'}/q_{M'}\in P_{\I'}$, or 
equivalently $\chi^{R'}\in q_{M'}P_{\I'}$. 

Finally, consider the case where a matroid $M'=(N',\mathcal{I}')\in\M$ gets dropped in
step~\ref{item:dropMat}. We need to show that 
$\chi^{R'}\vert_{N'}\in q_{M'}\cdot P_{\mathcal{I'}}$. 
Let $x^*$ be the optimal solution used in the algorithm when $M'$ was dropped. 
We have $|N'| - x^*(N') < q_{M'}$, and since $x^*(N')=r'(N')$, $|N'|$, and $q_{M'}$ are
integral, this implies $|N'| - x^*(N') \leq q_{M'}-1$.
So $N'$ can be partitioned into a basis of $M'$, which has size 
$r'(N')=x^*(N')\geq |N'|-(q_{M'}-1)$, and at most $q_{M'}-1$ other singleton sets. Each
singleton $\{e\}$ is independent in $M'$, since $0<x^*(e)\leq r'(\{e\})$ as
$x^*\vert_{N'}\in P_{\I'}$. Therefore, $N'$ can be partitioned into at most $q_{M'}$ independent
sets of $M'$. Intersecting these sets with $R'$ shows that $R'\cap N'$ can be
partitioned into at most $q_{M'}$ independent sets of $M'$.
%
\qed
\end{proof}

We now prove that the algorithm terminates. Note that refinements guarantee that
whenever the algorithm is at step~\ref{item:dropMat}, then for any $M'=(N',\I')\in\M$,
only the constraint of $P_{\mathcal{I}_i}$ corresponding to $N'$ may be $x^*$-tight. 
This allows us to leverage ideas similar to those
in~\cite{bansal_2009_additive,kiraly_2008_degree} 
to show that step~\ref{item:dropMat} is well defined.

\begin{lemma} \label{lem:terminate}
Algorithm~\ref{alg:mainAlg} terminates in at most $(2k+1)|N|$ iterations.
\end{lemma}

\begin{proof}
We show that whenever the algorithm is at step~\ref{item:dropMat}, then at least one
matroid in our collection can be dropped. This implies the above bound on the number
of iterations as follows. There can be at most $|N|$ deletions or contractions.
Each
matroid $M_i=(N_i,\I_i)$ in our input spawns at most $|N_i|$ refinements, as each
refinement of a matroid 
creates two
matroids with disjoint (nonempty) ground sets. This also means that step~\ref{item:dropMat} can be executed at most $k|N|$ times.

We focus on showing that step~\ref{item:dropMat} is well defined.
Consider the current collection of matroids $\M$. (Recall that $\M$ does not contain the
current version of $M_0$.) 
Let $x^*$ be the current basic solution, which is not integral; otherwise every
element would have been deleted or contracted in step~\ref{item:delcontr} and we would
have terminated in step~\ref{item:exit1}. Since we deleted all elements $e$ with
$x^*(e)=0$, the current ground set $N$ satisfies $N = \supp(x^*)$. 

Consider a full-rank subsystem of~\eqref{eq:mainLP}, $Ax=b$, consisting of linearly
independent, $x^*$-tight constraints. By standard uncrossing arguments, we may assume
that the constraints of $Ax=b$ coming from a single matroid correspond to a nested
family of sets. 
The system $Ax=b$ must contain some constraint corresponding to a
matroid $M'\in\M$. Otherwise, we would have a full-rank system consisting of constraints
coming from only one matroid, namely $M_0$, which would yield a unique integral solution; but $x^*$ is
not integral. 
Furthermore, for a matroid $M'=(N',\I')\in\M$, the only constraint of $P_{\mathcal{I}'}$
that can be $x^*$-tight corresponds to $N'$, as
otherwise, $M'$ would have been refined in step~\ref{item:refine}. So a matroid $M'\in\M$  
gives rise to at most one row of $A$, which we denote by $A_{M'}$ if it exists.
Let $\es\subsetneq S_1\subsetneq\ldots\subsetneq S_p\sse N_0=N$ denote the nested family of
sets that give rise to the constraints of $M_0$ in our full-rank system.

Consider the following token-counting argument. Each $e\in N$ gives $x^*(e)$ tokens to the 
row of $A$ corresponding to the smallest set $S_\ell$ containing $e$ (if one exists).
It also supplies $\bigl(1-x^*(e)\bigr)/q_{M'}$ tokens to every row $A_{M'}$ corresponding
to a matroid $M'\in\M$ whose ground set contains $e$.
%
Since the $q$-values satisfy \eqref{eq:loadProp}, every $e\in N$ supplies at most one
unit of token in total to the rows of $A$. Every row of $A$ corresponding to a set
$S_\ell$ receives $x^*(S_\ell)-x^*(S_{\ell-1})$ tokens, where $S_0\coloneqq\es$. This is
positive and integer, and thus at least $1$.
We claim that there is some $e\in N$ that supplies strictly less than one token unit. 
Given this, it must be that there is a row $A_{M'}$ corresponding to a matroid
$M'=(N',\I')\in\M$ that receives less than $1$ token-unit; thus 
$|N'|-x^*(N')<q_{M'}$ as desired.

Finally, we prove the claim. If every element supplies {\em exactly} one token-unit, then 
it must be that: 
\begin{enumerate*}
\item $S_p=N$, 
\item inequality \eqref{eq:loadProp} is tight for all $e\in N$, and
\item for every $e\in N$, every matroid $M'=(N',\I')\in\M$ with $e\in N'$ gives rise
to a row $A_{M'}$.
\end{enumerate*}
But then $\sum_{M'\in\M}\frac{1}{q_{M'}}\cdot A_{M'}=\chi^N$, which is the row of $A$
corresponding to the constraint of $M_0$ for the set $S_p$. This contradicts that $A$ has
full rank. \qed
\end{proof}

%% file: applications.tex
\section{Further applications and extensions} \label{apps}

\subsection{Generalized Matroidal Degree-Bounded Spanning Trees} \label{section:matroidal_MST}

We introduce the generalized matroidal degree-bounded spanning tree problem (\GMDBMST): given an undirected graph 
$G=(V,E)$ with edge costs $c\in\R^E$, disjoint node-sets $S_1,\ldots,S_k$, and
matroids $M_i=(\dt(S_i),\I_i)\;\forall i=[k]$, we want to find a spanning tree $T$ of minimum cost such that
$T\cap\dt(S_i)\in\I_i\;\forall i\in [k]$. Here, $\dt(S_i)$ is the set of edges of 
$G$ that cross $S_i$. This generalizes matroidal
degree-bounded MSTs considered by~\cite{zenklusen_2012_matroidal}, 
wherein each node $\{v\}$ is an $S_i$ set.

Using Theorem~\ref{thm:mainthm}, we obtain an algorithm for \GMDBMST{} that violates the matroid constraints by a factor of at most 2. Let $\OPT$ denote the optimal value for this problem.

\begin{theorem}
    There exists an efficient algorithm for \GMDBMST{} that computes a spanning tree $T \subseteq E$ such that $c(T) \le \OPT$ and $T \cap \delta(S_i)$ is $2$-independent in $M_i$ for every $i \in [k]$.
\end{theorem}

\begin{proof}
Let $M_0$ be the graphic matroid of $G$. We apply Theorem~\ref{thm:mainthm} to the matroids $\{M_i\}_{i \in \{0, \dots, k\}}$, with weight function $w := -c$, setting $q_i = 2$ for all $i \in [k]$. Note that these $q$-values satisfy \eqref{eq:loadProp}, since each edge belongs to the boundary of at most two sets $S_i$.

Note that \eqref{eq:mainLP} is a relaxation of our instance of \GMDBMST{} (with a flipped-sign objective). If \eqref{eq:mainLP} is infeasible, then so is our instance. Otherwise, we obtain a set $T$ of edges such that $c(T) = -w(T) \le - \OPT_\text{\ref{eq:mainLP}} \le \OPT$, and such that $T \cap \delta(S_i)$ is $2$-independent in $M_i$ for every $i \in [k]$. Further, $T$ is a spanning tree of $G$ since it is a basis of $M_0$. \qed
\end{proof}

Whereas~\cite{zenklusen_2012_matroidal} obtains an $O(1)$ {\em additive}
violation of the matroid constraints for the matroidal degree-bounded MST
problem, 
we show in Appendix~\ref{sec:largeAdditiveViolation} (see Theorem~\ref{thm:hardness_matroidal_mst}) that, unless \p$=$\np,
it is impossible to obtain such an additive guarantee for \gmdst. 
This follows from the same replication idea used to rule out small additive
violations for \eqref{eq:baseQMatchoid}.

\subsection{Extension to knapsack constraints}

We can consider a generalization of \eqref{eq:baseQMatchoid}, where, in addition to
the matroids 
$M_0, \dots, M_k$ (over subsets of $N$) and the weight vector $w\in\R^N$, we have $t$ knapsack 
constraints, indexed by $i = k + 1, \dots, k + t$. The $i$-th knapsack constraint is
specified by a ground set $N_i\sse N$, a cost vector $c^i\in\mathbb{R}^{N_i}_{\geq 0}$, and a budget
$U_i\geq 0$.  
The goal is to find a maximum-weight set $R$ such that
$R\in\base_0\cap\I_1\cap\ldots\cap\I_k$, and satisfying $c^i(R\cap N_i)\leq U_i$ for all
$i=k+1,\ldots,k+t$. 

We assume without loss of generality that $\max_{e \in N_i} c^i_e \le U_i$ for $i = k + 1, \dots, k+t$. (If this inequality did not hold, then we could identify and drop some elements that do not belong to any feasible solution.)

We obtain the problem
\begin{align*}
\max\ \Bigl\{w(I): \ \ & I \in \base_0, \quad I \in \mathcal{I}_i\ \ \forall i\in[k], \\
& c^i(I \cap N_i) \le U_i\ \ \forall i=k+1, \ldots, k+t\quad \Bigr\}, 
\end{align*}
and its LP-relaxation
\begin{align*}\label{eq:knapLP}
\max\ \Bigl\{w^Tx: \ \ & x\in\R_{\ge 0}^N, \quad x\in P_{\base_0}, \quad
x\vert_{N_i}\in P_{\I_i}\ \ \forall i\in[k], \\
& (c^i)^Tx\vert_{N_i} \le U_i\ \ \forall i=k+1, \ldots, k+t\quad \Bigr\}.
\tag{\LPKnapsack}
\end{align*}
We obtain the following generalization of Theorem~\ref{thm:mainthm}.

\begin{theorem} \label{thm:knapThm}
Let $q_1,\ldots, q_{k + t} \in \mathbb{Z}_{\geq 1}$ such that
\begin{equation}\label{eq:loadPropKnap}
\sum_{i\in [k + t]: e\in N_i} q_i^{-1} \leq 1 \qquad \forall e\in N\espace.
\end{equation}
If \eqref{eq:knapLP} is feasible, then one can efficiently compute $R \subseteq N$ such that
\begin{enumerate}
\item $R \in \base_0$; \label{prop_knap_first}
\item $w(R)\geq\OPT_{\text{\ref{eq:knapLP}}}$; 
\item $R$ is $q_i$-independent in $M_i$ for all $i \in [k]$; and \label{prop_knap_second_to_last} 
\item $c^i(R \cap N_i) \le U_i + q_i \cdot \bigl(\max_{e\in N_i}c^i_e\bigr) \le (q_i + 1) U_i$ for all $i \in \{k + 1, \dots, k + t\}$. \label{prop_knap_last}
\noindent 
\end{enumerate}
\end{theorem}
Again, taking $q_i=\Delta\coloneqq \max_{e\in N}\bigl|\{j\in[k + t]: e\in N_j\}\bigr|$ for all $i\in[k + t]$ satisfies \eqref{eq:loadPropKnap}.
So we can obtain 
$\Dt$-independence for all the matroid constraints, and violate each knapsack constraint by at most a factor of $\Delta + 1$.

We prove Theorem~\ref{thm:knapThm} utilizing Algorithm~\ref{alg:knapAlg}, obtained via a small modification of Algorithm~\ref{alg:mainAlg}. The result follows from Lemmas~\ref{lem:knapterminate} and \ref{lem:knapapprox}.

\begin{algorithm2e}[h]
\begin{enumerate}[nosep,leftmargin=0.5em,label=\arabic*.,ref=\arabic*]
    \item\label{item_knap:init} Initialize $\M\assign\{M_1,\ldots,M_k\}$, $q_{M_i} \gets q_i$ for all $i\in [k]$, $\mathcal{K}\assign\{k + 1, \dots, k + t\}$.
\item\label{item_knap:solveLP} Compute an optimal basic
solution $x^*$ to~\eqref{eq:knapLP} for the matroids $\{M_0\}\cup\M$ and the knapsack constraints indexed by $\mathcal{K}$.

\item\label{item_knap:delcontr} 
Delete all $e\in N$ with $x^*(e)=0$ and contract all $e \in N$ with $x^*(e)=1$ from all relevant matroids, updating also the ground set $N$. For every $i \in \mathcal{K}$, update $U_i \gets U_i - c^i(\{e \in N_i : x^*(e)=1\})$ and $N_i \gets N_i \setminus \{e \in N_i : x^*(e)=1\}$. 

\item\label{item_knap:exit1}
If $N=\emptyset$: \textbf{return} the set of all elements contracted so far.

\item\label{item_knap:refine} \textbf{While} there is a matroid $M'=(N',\I')\in\M$ 
with associated rank runction $r'$, s.t. $\exists\es\neq S\subsetneq N'$ with
$x^*(S)=r'(S)$: 
\hspace*{0.5cm}\parbox[t]{0.9\linewidth}{
    \emph{(Refinement.)} Set $M_1' = M'|_S$, $M_2' = M' / S$, and $q_{M_1'} = q_{M_2'} = q_{M'}$.\\
    Update $\M\assign(\M\sm\{M'\})\cup\{M'_1, M'_2 \}$.
}

\item\label{item_knap:drop}
If there exists a matroid $M'=(N',\I')\in\M$ with associated rank function $r'$, 
such that $x^*(N')=r'(N')$ and $|N'| - x^*(N')<q_{M'}$, then remove $M'$ from $\M$. Otherwise, find $i \in \mathcal{K}$ such that $|N_i| - x^*(N_i) \le q_{i}$; remove $i$ from $\mathcal{K}$.
Go to step~\ref{item:solveLP}.
\end{enumerate}


\caption{Iterative refinement/relaxation algorithm
for Theorem~\ref{thm:knapThm}}{}
\label{alg:knapAlg}
\end{algorithm2e}

\begin{lemma} \label{lem:knapapprox}
Suppose that Algorithm~\ref{alg:knapAlg} returns a set $R\subseteq N$. Then, $R$ satisfies
the properties stated in Theorem~\ref{thm:knapThm}. 
\end{lemma}

\begin{proof}
Parts \ref{prop_knap_first}--\ref{prop_knap_second_to_last} follow from the arguments used in the proof of Lemma~\ref{lem:approx}. To prove part \ref{prop_knap_last}, consider the $i$-th knapsack
constraint. Note that the only place where we possibly introduce a violation in the
knapsack constraint is when we drop the constraint. If $x^*$ is the optimal solution just
before we drop the constraint, then we know that $(c^i)^Tx^*\vert_{N_i}\leq U_i$. (Note
that $N_i$ and $U_i$ refer to the updated ground set and budget.) 
It follows that if $S$ denotes the set of elements included from this residual ground set
$N_i$, then the additive violation in the knapsack constraint is at most
\begin{align*}
c^i(S)-U_i & \leq c^i(N_i)-U_i = (c^i)^Tx^*\vert_{N_i} + \sum_{e \in N_i} c^i_e(1 - x^*_e) - U_i \\
& \le \bigl(\max_{e\in N_i}c^i_e\bigr) \sum_{e \in N_i} (1 - x^*_e) = \bigl(\max_{e\in N_i}c^i_e\bigr) \bigl(|N_i|-x^*(N_i)\bigr) \\
& \le q_i\cdot\bigl(\max_{e\in N_i}c^i_e\bigr) \le q_i \cdot U_i \enspace.\qedhere
\end{align*}
\end{proof}

\begin{lemma} \label{lem:knapterminate}
Algorithm~\ref{alg:knapAlg} terminates in at most $(2k+1)|N|+t$ iterations.
\end{lemma}

\begin{proof}
We claim that whenever the algorithm is at step~\ref{item_knap:drop}, there is at least one matroid constraint, or knapsack constraint that can be dropped. Assuming this, it is clear that we drop a knapsack constraint at most $t$ times. The number of the remaining types of operations can be bounded by $(2k+1)|N|$, as explained in the proof of Lemma~\ref{lem:terminate}; the bound in the lemma statement follows.

It remains to prove the claim, which
follows from the token-counting argument used in the proof of Lemma~\ref{lem:terminate}.
Recall that if $Ax=b$ is a full-rank subsystem of \eqref{eq:knapLP} consisting of linearly
independent $x^*$-tight constraints, then we may assume that the rows of $A$ corresponding
to the $M_0$-constraints form a nested family $\C$.
We define a token-assignment scheme, where each $e\in N$ supplies $x^*(e)$ tokens to the
row of $A$ corresponding to the smallest set in $\C$ containing $e$ (if one exists),
and $\bigl(1-x^*(e)\bigr)/q_{M'}$ to each row $A_{M'}$
coming from a matroid $M'\in\M$ in our collection whose ground set contains $e$. 
{\em Additionally}, every $e\in N$ now also supplies $\bigl(1-x^*(e)\bigr)/q_i$ tokens to
each row of $A$ originating from a knapsack constraint whose ground set contains $e$.
Under this scheme, as before, given the constraint on our $q$-values, it follows that every
$e\in N$ supplies at most $1$ token unit. Also, as before, each row of $A$ corresponding
to an $M_0$ constraint receives at least $1$ token unit. So either there is some row
$A_{M'}$ coming from a matroid in $\M$ that receives strictly less than $1$ token-unit, 
or there must be some row of $A$ corresponding to a knapsack constraint that receives at
most $1$ token-unit; the latter case corresponds to a knapsack constraint $i$ with
$|N_i|-x^*(N_i)\leq q_i$. \qed
\end{proof}

%% file: acknowledgments.tex
\section*{Acknowledgments}

We are thankful to Lap Chi Lau for pointing us to relevant literature.

%% file: largeAdditiveViolation.tex
\section{Impossibility of achieving small additive violations}
\label{sec:largeAdditiveViolation} \label{addviol}
We show that Theorem~\ref{thm:mainThm} for problem \eqref{eq:baseQMatchoid} cannot be
strengthened to yield a basis of $M_0$ that has small additive violation for
the matroid constraints of $M_1,\ldots,M_k$, even when $k=2$.

We first define additive violation precisely.
Given a matroid $M = (N,\mathcal{I})$ with rank function $r$, we say that a set 
$R \subseteq N$ is {\em $\mu$-additively independent} in $M$ if $|R|-r(R) \le\mu$;
equivalently, we can remove at most $\mu$ elements from $R$ to obtain an independent set in $M$.
Unlike results for degree-bounded spanning trees, or matroidal degree-bounded
MST~\cite{zenklusen_2012_matroidal}, 
we show that small additive violation is not possible in polytime (assuming \p$\neq$\np) even
for the special case of \eqref{eq:baseQMatchoid} where $k=2$, 
so we seek a basis of $M_0$ that is independent in $M_1,M_2$. 

\begin{theorem} \label{theorem:additive_violation} \label{hardness}
Let $f(n) = O(n^{1-\varepsilon})$, where $\varepsilon > 0$ is a constant. 
Suppose we have a polytime algorithm $\A$ for \eqref{eq:baseQMatchoid} 
that returns a basis of $M_0$ that is $f(|N|)$-additively independent in $M_i$ for $i=1,2$. Then we
can find in polytime a basis of $M_0$ that is independent in $M_1,M_2$.
\end{theorem}

The problem of finding a basis of $M_0$ that is independent in $M_1,M_2$ is \nphard,
as shown by an easy reduction from the directed Hamiltonian path problem. 
Thus, Theorem~\ref{hardness} shows that it is {\nphard} to obtain an additive violation
for problem \eqref{eq:baseQMatchoid} that is substantially better than linear violation.
\nnote{Again, fix sublinear; made an attempt.}

\begin{proofof}{Theorem~\ref{hardness}}
Choose $t$ large enough so that $t > 2f(t|N|)$. Since 
$f(n) = O\bigl(n^{1-\varepsilon}\bigr)$, this is achieved by some $t=\poly(|N|)$.
For each $i \in \{0, 1, 2\}$, let $M'_i$ be the direct sum of $t$ copies of $M_i$.
Let $N'$ be the ground set of these matroids, which consists of $t$ disjoint copies of
$N$, which we label $N_1,\ldots,N_t$.

Clearly, the instance $(M'_0, M'_1, M'_2)$ is feasible iff the original instance is
feasible. Suppose that running $\A$ on the replicated instance yields a basis $R'$ of
$M'_0$ that has the stated additive violation for the matroids $M'_1, M'_2$.
Hence, there are two sets $Q_1,Q_2\subseteq R'$ with $|Q_1|,|Q_2|\leq f(t|N|)$, such that $R'\setminus Q_i$ is independent in $M'_i$ for $i = 1, 2$. Hence, $R'\setminus (Q_1\cup Q_2)$ is independent in both $M'_1$ and $M'_2$. Because $|Q_1\cup Q_2| \leq 2f(t|N|) < t$, we have by the pigeonhole principle that there is one $j\in [t]$ such that $(Q_1\cup Q_2)\cap N_j=\emptyset$. This implies that $R = R'\cap N_j = (R'\setminus (Q_1\cup Q_2))\cap N_j$, when interpreted on the ground set $N$, is independent in both $M_1$ and $M_2$. Moreover, the elements of $R$, when interpreted on the ground set $N$, are a basis in $M_0$ because $R'$ is a basis in $M'_0$. Hence, $R$ is the desired basis without any violations.
\end{proofof}

We can mimic the above proof to show that one cannot achieve small additive violations to
the matroid constraints for another special case of~\eqref{eq:baseQMatchoid}, namely,  
the problem \GMDBMST introduced in Section~\ref{apps}:
we are given
$\bigl(G=(V,E),c\in\R^E,\{M_i=(\dt(S_i),\I_i)\}_{i\in[k]}\bigr)$, where the $S_i$s are
pairwise-disjoint node sets, and we seek a min-cost spanning tree $T$ such that
$T\cap\dt(S_i)\in\I_i$ form all $i \in [k]$. 
As above, we show that if we can find a spanning tree that has small additive violation
for the matroids $M_1,\ldots,M_k$, then we can find a feasible solution to \gmdst. The
latter is \nphard, even when the $M_i$s are uniform matroids, as this captures the
degree-bounded spanning tree, and hence the Hamiltonian path problem.

This provides an alternative proof of why one cannot achieve additive guarantees for
\eqref{eq:baseQMatchoid}. (Note, however, that the hardness results from Theorems~\ref{hardness} and
\ref{thm:hardness_matroidal_mst} are orthogonal.)

\begin{theorem} \label{thm:hardness_matroidal_mst}
Let $f(n) = O(n^{1-\varepsilon})$, where $\varepsilon > 0$ is a constant.  Suppose we have
a polytime algorithm for \gmdst that returns a spanning tree $T \subseteq E$ such that
$T \cap \delta(S_i)$ is $f(|E|)$-additively independent in $M_i$ for all $i\in[k]$. Then
we can find a feasible solution to \gmdst in polytime. 
\end{theorem}

\begin{proof} 
Let $t = \poly(|E|)$ be such that $t>kf\bigl(t(|E| + 1)\bigr)$. We construct a graph by
taking the union of $t$ copies of $G$. To connect the copies with each other, we utilize
an additional vertex $z$, connected by an edge to all the copies of an arbitrarily chosen
vertex of $G$. Let $G'=(V',E')$ denote the graph thus obtained, and note that
$|E'|=t(|E|+1)$. 
For each $i\in[k]$, let $S'_i$ be the union of all copies of $S_i$. The matroid $M'_i$ on 
$\dt_{G'}(S'_i)$ is the union of the direct sum of $t$ copies of $M_i$, and the free
matroid on the edges of $\dt_{G'}(S'_i)$ incident to $z$.

Note that the resulting \gmdst instance is feasible iff the original \gmdst instance is
feasible. Furthermore, if $T'$ is a spanning tree of $G'$, then $T'$ restricted to each
copy of $G$ yields a spanning tree of $G$. The choice of $t$ ensures that in some copy, the
resulting tree satisfies the matroid constraints for $M_1,\ldots,M_k$, and is therefore a
feasible solution to the original \gmdst instance.
\qed
\end{proof}

%% file: independenceVersion.tex
\begin{proofof}{Corollary~\ref{cor:mainInd}}
    Extend $N$ by adding a set $F$ of $r(N_0)$ additional elements with $0$ weight, where $r$ is the rank function of $M_0$.
We modify $M_0$ to  a matroid $\widehat{M}_0$ on the ground set $N_0 \cup F$, given by the 
rank function 
    $\widehat{r}(S)\coloneqq\min\{r(S\cap N_0) + |S\cap F|, r(N_0)\}$.
    That is, $\widehat{M}_0$ is the union of $M_0$ with a free matroid on $F$, but then
    truncated to have rank $r(N_0)$. 
    It is now easy to see that if $x\in\R^{N\cup F}$ lies in $P_{\widehat{\base}_0}$, then
    $x\vert_{N_0}\in P_{\I_0}$. Moreover, we can extend $x\in\R^N$ with $x\in P_{\I_0}$ to 
    $x'\in\R^{N\cup F}$ so that $x'\vert_{N_0\cup F}\in P_{\widehat{\base}_0}$ and
    $x'\vert_N=x$.
    The corollary thus follows by applying Theorem~\ref{thm:mainThm} to $\widehat{M}_0, M_1, \ldots, M_k$. 
\end{proofof}

%% file: appendix.tex

\begin{proofof}{Lemma~\ref{lem:refinement}}
%
For any set $A\sse N$, we have 
$$
x(A)=x(A\cap S)+x(A\sm S)\leq r_1(A\cap S)+r_2(A\sm S)=r(A\cap S)+r(A\cup S)-r(S).
$$
Using submodularity of $r$, we have $r(A\cup S)-r(S)\leq r(A)-r(A\cap S)$, so 
$x(A)\leq r(A)$. This holds for every $A\sse N$, so $x\in P_{\I}$. 
\end{proofof}

\begin{proofof}{Lemma~\ref{lem:remFeasible}}
Let $N_1=S$ and $N_2=N\sm S$.
For $i\in \{1,2\}$, to show $x\vert_{N_i}\in P_{\mathcal{I}_i}$ we have to verify that
$x\vert_{N_i}$ fulfills all constraints of the matroid polytope $P_{\mathcal{I}_i} =
\{y\in \mathbb{R}_{\geq 0}^{N_i} : y(Q) \leq r_i(Q) \;\forall Q\subseteq N_i \}$. For
$i=1$ this immediately follows from the fact that $r_1$ is the restriction of $r$ to
subsets of $S$, 
and $x\in P_{\mathcal{I}}$; for any $A\subseteq S$, we have 
\begin{equation*}
x\vert_{N_1}(A) = x(A) \leq r(A) = r_1(A)\espace,
\end{equation*}
where the inequality follows from $x\in P_{\mathcal{I}}$, and the second equation from~\eqref{eq:rankRefinements}.

Moreover, $x\vert_{N_2}\in P_{\mathcal{I}_2}$ holds because for any $B\subseteq N\setminus S$, we have
\begin{align*}
x(B) = x(B) + x(S) - r(S)
     = x(B\cup S) - r(S)
     \leq r(B\cup S) - r(S)
     = r_2(B)\espace,
\end{align*}
where the first equation is a consequence of $x(S) = r(S)$, the inequality is implied by
$x\in P_{\mathcal{I}}$, and the last equation holds due to~\eqref{eq:rankRefinements}.
\end{proofof}